\pdfoutput=1

\documentclass[a4paper,USenglish,cleveref, autoref, thm-restate, numberwithinsect]{lipics-v2021}
\title{Local Access to Random Walks}

\titlerunning{Local Access to Random Walks}

\author{Amartya Shankha Biswas}{CSAIL, MIT, Cambridge MA, USA \and asbiswas@mit.edu }{}{}{Big George Ventures Fund, MIT-IBM Watson AI Lab and Research Collaboration Agreement No. W1771646, NSF awards CCF-1733808 and IIS-1741137}

\author{Edward Pyne}{Harvard University, Cambridge MA, USA \and epyne@college.harvard.edu }{}{}{}

\author{Ronitt Rubinfeld}{CSAIL, MIT, Cambridge MA, USA \and ronitt@mit.edu }{}{}{NSF awards CCF-2006664, CCF-1740751, IIS-1741137, Fintech@CSAIL. Part of this work was done while the author was participating in the program
on Probability, Geometry, and Computation in High Dimensions at the Simons Institute for the Theory of Computing.}

\authorrunning{A.\,S. Biswas, E. Pyne and R. Rubinfeld}

\Copyright{Amartya Shankha Biswas, Edward Pyne and Ronitt Rubinfeld}



\category{} %optional, e.g. invited paper

\relatedversion{} %optional, e.g. full version hosted on arXiv, HAL, or other respository/website
\acknowledgements{The second author thanks Andrew Lu for valuable comments on a draft of this paper.}

\nolinenumbers

\hideLIPIcs  

\EventEditors{John Q. Open and Joan R. Access}
\EventNoEds{2}
\EventLongTitle{42nd Conference on Very Important Topics (CVIT 2016)}
\EventShortTitle{CVIT 2016}
\EventAcronym{CVIT}
\EventYear{2016}
\EventDate{December 24--27, 2016}
\EventLocation{Little Whinging, United Kingdom}
\EventLogo{}
\SeriesVolume{42}
\ArticleNo{23}
\usepackage{dirtytalk}
\usepackage{mathtools}
\usepackage{todonotes}
\usepackage{thm-restate}

\newcommand{\N}{\mathbb{N}}
\newcommand{\R}{\mathbb{R}}
\newcommand{\Z}{\mathbb{Z}}
\newcommand{\E}{\mathbb{E}}
\newcommand{\I}{\mathbb{I}}
\newcommand{\ep}{\epsilon}
\newcommand{\ra}{\rightarrow}
\newcommand{\la}{\leftarrow}

\newcommand{\F}{F_{G}}
\newcommand{\A}{\mathcal{A}}

\newcommand{\D}{D_{G,\A,Q}}
\newcommand{\Q}{Q}
\DeclareMathOperator{\Cay}{Cay}
\newcommand{\U}[1][ ]{U_{G}^{#1}}

\DeclareMathOperator{\poly}{poly}
\DeclareMathOperator{\polylog}{polylog}

\DeclareMathOperator*{\argmin}{arg\,min}
\DeclareMathOperator{\MNom}{MNom}
\DeclareMathOperator{\MHGeom}{MHGeom}
\DeclareMathOperator{\BNom}{BNom}
\DeclareMathOperator{\HGeom}{HGeom}
\DeclareMathOperator{\PL}{PL}
\newcommand{\dist}{\textsc{d}}
\renewcommand{\sp}{\textsc{SP}}

\newcommand{\Gd}{\mathbf{G}(n,d)}
\newcommand{\Eg}{\E_{G\la \Gd}}

\newcommand{\rn}{\textsc{rand\_neighbor}}
\newcommand{\rv}{\textsc{rand\_vertex}}
\newcommand{\rp}{\textsc{rand\_path}}
\newcommand{\pt}{\textsc{position}}

\newtheorem{notation}[theorem]{Notation}

\begin{document}
\maketitle
\begin{abstract}
For a graph $G$ on $n$ vertices, naively sampling the position of a random walk of at time $t$ requires work $\Omega(t)$.
We desire \emph{local access} algorithms supporting $\pt(G,s,t)$ queries, which return the position of a random walk from some start vertex $s$ at time $t$,
where the joint distribution of returned positions is $1/\poly(n)$ close to the uniform distribution over such walks in $\ell_1$ distance.

We first give an algorithm for local access to walks on undirected regular graphs with $\widetilde{O}(\frac{1}{1-\lambda}\sqrt{n})$ runtime per query,
where $\lambda$ is the second-largest eigenvalue in absolute value.  Since random $d$-regular graphs are expanders with high probability,
this gives an $\widetilde{O}(\sqrt{n})$ algorithm for $G(n,d)$, which improves on the naive method for small numbers of queries.

We then prove that no that algorithm with subconstant error given probe access to random $d$-regular graphs can have runtime better than $\Omega(\sqrt{n}/\log(n))$ per query in expectation, obtaining a nearly matching lower bound. We further show an $\Omega(n^{1/4})$ runtime per query lower bound even with an oblivious adversary (i.e. when the query sequence is fixed in advance).

We then show that for families of graphs with additional group theoretic structure, dramatically better results can be achieved.
We give local access to walks on small-degree abelian Cayley graphs, including cycles and hypercubes, with runtime $\text{polylog}(n)$ per query. This also allows for efficient local access to walks on $\polylog$ degree expanders. We extend our results to graphs constructed using the tensor product (giving local access to walks on degree $n^\ep$ graphs for any $\ep \in (0,1]$) and Cartesian product.
\end{abstract}

\section{Introduction}
Given some huge random object that an algorithm would like to query, is it always necessary to generate the entire object up front?
For sublinear time algorithms, generating such a large object would dominate the runtime.
Recent works~\cite{BRY, GGN, NN, ELMR} demonstrated this was not always necessary, giving incremental query access to random objects such as random graphs, Dyck paths and graph colorings. These {\em local access} algorithms answer queries in a manner consistent with an instance of the random object sampled from the true distribution (or close to it).

In this work we explore the question of implementing local access to random walks.
Given a graph $G$ on $n$ vertices, taking a random walk of length $t$ requires time $\Omega(t)$.
Random walks are a critical primitive in many algorithms~\cite{MST,KM,FGRV}, including sublinear ones~\cite{GR,RV,AKP}.
But since $t$ can be large, one may want to generate only the required segments of the walk that are needed at the present time,
while ensuring the joint distribution of the returned segments is close to the true distribution of random walks.

As is common in the setting of sublinear and local algorithms,
we assume that we are given access to a graph $G$ on $n$ vertices through query oracles.
This allows us to work with graphs that are too large to fit in main memory,
and also results in running times that are not dominated by the size of the input.
Our goal is to implement {$\pt(G,s,t)$} queries, which return the position of a random walk starting from vertex $s$ at time $t$,
such that, given a sequence of queries, the joint distribution of returned positions is $1/\poly(n)$-close
to the true uniform distribution of those positions over random walks (in $\ell_1$ distance).
We desire per query runtime that is sublinear in $n$ and $t$, and preferably polylogarithmic in both.
In that case, locally generating \emph{all} vertices in a walk of length $t$ (in an arbitrary order) has total work within a polylog factor of the naive runtime.

Obtaining efficient random access for arbitrary graphs without knowing the entire structure seems to be a very difficult problem,
and therefore in this paper we restrict our attention to regular graphs. However, regular graphs include widely studied families such as random regular and Cayley graphs, both of which which we analyze.

\subsection{Our Results and Techniques}
\label{sec:our_results_and_techniques}

We begin by presenting a $\tilde{\mathcal O}(\frac{1}{1-\lambda}\sqrt{n})$ algorithm
that provides local access to undirected $d$ regular graphs with spectral expansion $\lambda$.
This algorithm maintains a collection of \emph{revealed} time values and the associated positions of the random walk,
where the revealed positions are a superset of the queried positions,
specifically the positions that were either queried directly, or were \say{determined} by the local access algorithm in order to answer a query.
The key idea in this algorithm is to handle queries in three different ways, based on the queried time relative to all other revealed times. The algorithm maintains the invariant that all determined times are either directly adjacent or separated by at least twice the mixing time.
Given query time $t$, if the new query time is further than twice the mixing time away from any other revealed time,
the algorithm simply returns a random vertex as the corresponding position and appends $t$ to the list of revealed locations.
Alternately, if $t$ is close to exactly one revealed time (either smaller or greater) and far from the other one,
then the algorithm simply simulates the entire walk between $t$ and the closer revealed time,
which takes $\tilde{\mathcal O}(\frac{1}{1-\lambda})$ steps.
Finally, the most interesting case is when there are two revealed positions on either side that are close to $t$, but not too close.
In this setting, the algorithm samples $\tilde{\mathcal O}(\sqrt{n})$ random walks from both revealed locations,
each of length half the interval, until a collision is found. Since walks of this length are well mixed, a collision occurs with high probability. 
The two colliding walks are then stitched together to interpolate the walk between the nearest revealed locations,
and then the corresponding position at time $t$ can be returned.

Moving forwards, we demonstrate that such a runtime is optimal in general. Specifically, our lower bound holds for the case of \emph{random} $d$-regular graphs, which provides some evidence that obtaining fast query algorithms for \say{large} classes is challenging.
Our lower bounds present adaptively chosen query sequences, and demonstrate that for the vast majority of these random graphs,
any algorithm making $\tilde{\mathcal O}(\sqrt{n}/\log n)$ \emph{random-neighbor} and \emph{random-vertex} probes to the underlying graph $G$
will fail to answer the queries in a consistent manner.
The main structural result used here is Lemma~\ref{lem:randstruct} which states that as long as the algorithm makes fewer than $\Theta(\sqrt{n})$ probes,
the revealed edges and vertices of the graph will form a forest, and additionally, no trees will ever be merged, with probability at least $0.995$.
This allows us to define a distance metric $d(\cdot, \cdot)$ where $d(u, v)$ is the distance between vertices $u$ and $v$
\emph{using only the revealed edges}, and is defined to be $\infty$ if no such path has been revealed.
The high level strategy in the lower bound is to first query the positions $v_0, v_e$ of the walk at time $t=0$ and $t=\sqrt{n}$ respectively,
and then adaptively query $\mathcal O(\log n)$ intermediate positions (where the query times may depend on the internal state of the algorithm),
until an inconsistency is found.
The hypothesis at this point is that the algorithm does not actually know of a path of the correct length between the two returned vertices.
Specifically, we show that either the revealed edges fail to connect the vertices in the limited number of available probes,
or the known path between them is shorter than $\sqrt{n}/20$.
In the first case, we can perform binary search for a location such that we end up with two reported positions which are adjacent in time,
but do not have an edge between them, thus yielding the inconsistency.
The latter case is more complicated, and requires some case analysis, but we are able to query adaptively
and always find two positions $v_i$ and $v_j$ (revealed at times $t_i$ and $t_j$), such that one of the following two outcomes hold:
either the distance is too large $d(v_i, v_j) > |t_i - t_j|$ or the distance is too small $d(v_i, v_j) < |t_i - t_j|/2$.
In the first outcome, if the distance is greater, we can again perform binary search to find adjacent positions in the walk that are not connected by an edge.
For the second outcome, we again perform binary search to find a short segment with unusually short distance, and then query all intermediate locations
to find a segment of the walk $\sigma_1, \sigma_2,\cdots, \sigma_l$ of length $\Theta(\log n)$, such that $d(\sigma_1, \sigma_l) < l/2$.
Note that we are then able to query all the locations in this segment because its length was reduced to $\mathcal O(\log n)$.
The fact that this segment has much smaller distance than the time interval implies that there is a significant amount of backtracking,
and we demonstrate that the probability of significant backtracking over a truly random walk is $o(1)$.

We also prove an oblivious lower bound of $\Omega(n^{1/4})$, for the case when the queries do not depend on the internal state of our algorithm.
In this case, we present the sequence of query times $(n^{1/4}, 2, 3, \cdots, n^{1/4}-1)$.
If the algorithm makes $\mathcal O(n^{1/4})$ graph probes, then the total number of probes is bounded above by $\Theta(\sqrt{n})$,
and therefore we use the same structural Lemma~\ref{lem:randstruct} mentioned above in order to derive a contradiction.

Finally, motivated by the lack of efficient local access to walks on general classes of graphs,
we turn to algorithms for local access on families of graphs with additional algebraic structure.
We give \emph{efficient} local access to walks on small-degree abelian Cayley graphs (for instance, cycles and hypercubes).
This also allows for efficient local access to walks on a class of $\polylog$ degree expanders.
We extend our results to graphs constructed using the tensor product (giving local access to walks on degree $n^\ep$ graphs for any $\ep \in (0,1]$) and Cartesian product.

\subsection{Related Work}
\label{sec:related_work}
The problem of providing local access to huge random objects was first proposed in \cite{GGN, GGNf}.
Subsequent work in \cite{NN} presented algorithms that provide access to sparse Erdos-Renyi $G(n, p)$ graphs through \textsc{All-Neighbors} queries,
as long as the number of queries is small and $p = \mathcal O(\poly(\log n))$.
Many of the results in these earlier works only guarantee that the generated random objects \emph{appear} to look random,
as long as the number of queries are bounded, usually by $\mathcal O(\poly(\log n))$.
More recently, in \cite{ELMR}, an implementation of random recursive trees and BA preferential attachment graphs are presented. Further, local access is given for the \textsc{Next-Neighbor} query that returns the neighbors of a vertex in lexicographic order,
which is useful for accessing graphs where the degree is not bounded.
Subsequently, \cite{BRY} presented implementations for random $G(n, p)$ graphs for any value of $n$,
while supporting \textsc{Next-Neighbor} as well as the newly introduced \textsc{Random-Neighbor} queries.
In~\cite{BRY}, algorithms are provided for accessing random walks on the line, random Dyck paths, and random colorings of a graph.
Implementing access to random walks on the line graph was motivated by the implementation of interval summable functions in \cite{GGN, gil}.

\subsection{Organization}
In Section~\ref{sec:preliminaries} we introduce notation and basic sampling tools. In Section~\ref{sec:expansion} give a local access algorithm for undirected regular graphs with runtime in terms of expansion.
In Section~\ref{sec:regular}, we first apply the previous algorithm to random $d$-regular graphs. We then prove a nearly matching lower bound with respect to an adaptive or non-oblivious adversary
(one who has access to the internal state of our algorithm),
and a weaker bound with respect to an oblivious adversary.
In Section~\ref{sec:abelian} we give local access algorithms for small degree abelian Cayley graphs, such as hypercubes and cycles.
In Appendix~\ref{app:algebraic}, we give local access algorithms for the tensor and Cartesian graph products.

\section{Preliminaries}\label{sec:preliminaries}
We first define terminology and introduce basic tools for sampling. 
We characterize the closeness of query responses to true random walks via $\ell_1$ distance, and use $\ell_2$ distance for spectral arguments.
\begin{notation}~
\begin{itemize}
    \item Given distributions $A,B$ over a set $[S]$, the $\ell_1$ distance between $A$ and $B$ is defined as $||A-B||_1 = \sum_{i=1}^S |A_i-B_i|$. The $\ell_2$ distance is defined as $||A-B||_2 =\sqrt{\sum_{i=1}^S(A_i-B_i)^2}$.
    \item For some set $S$, let $U_{S}$ denote the uniform distribution over $S$. Let $s \la U_S$ be an element drawn from this distribution.
\end{itemize}
\end{notation}

Next, we define notation for the distribution of random walks on fixed graphs.
\begin{notation}
    Given regular $G=(V,E)$ where $V=[n]$, $v_1,v_2 \in V$ and $t \in \N$:
    \begin{itemize}
        \item Let $\lambda(G)=\max_{x \in \R^n:x\perp 1}||Wx||_2/||x||_2$ where $W$ is the random walk matrix of $G$.
        \item Let $D_C(G,v_1,v_2,t)$ be the distribution over random walks of length $t$ from $v_1$ that end at $v_2$. As $G$ is regular, this is the uniform distribution over all satisfying walks.
        \item Let $\U[\ell]$ be the unconditional distribution of random walks from vertex $1$ of length $\ell$.
        \item For any finite set of times $S \in \N^k$, let $P_S(\U)$ be the distribution of the positions at times $S$ of random walks from vertex 1. We will measure the accuracy of an algorithm given time queries $S$ by bounding the $\ell_1$ distance of its responses to $P_S(\U)$. For notational convenience, let $P_i=P_{\{i\}}$.
    \end{itemize}
\end{notation}

We can then define the class of algorithms we consider. 
\begin{definition}
A \textbf{local access algorithm} $\A$ for a graph $G$ is an algorithm that,
given $\ep>0$ and $B \in \N$ at initialization and a sequence of queries $T=t_1,\ldots,t_r$ for $r\leq B$,
returns vertices $(v_{t_1},\ldots,v_{t_r})\la D$ such that $||D-P_T(\U)||_1\leq \ep$.
\end{definition}
We believe the useful regime to be setting $B=n^c$ and $\ep=n^{-c-c'}$ for desired constants $c,c'$, giving a polynomial approximation in $\ell_1$ distance. Moreover, one can implicitly restrict all query times to be below some polynomial threshold for the remainder of the paper. 
\begin{definition}
    We say a local access algorithm $\A$ is \textbf{efficient} if for $\ep=1/\poly(n),$ $B=\poly(n)$ and $t_i\leq \poly(n)$ for all $i$, the algorithm runs in time $\polylog(n)$ per query. 
\end{definition}
Our definition of efficiency is motivated by the fact that taking a random walk of length $t$ requires time $\Omega(t)$, so an efficient algorithm allows one to incrementally construct a random walk in an arbitrary query order with total runtime within a $\polylog$ factor of the naive algorithm.

\begin{definition}
A sequence of queries $T$ to a \textbf{local access algorithm} is considered to be \textbf{adaptive} with respect to a local access algorithm $\A$,
if the $i^{th}$ query is allowed to depend on the \textit{internal} state of the algorithm after query $i-1$.
Additionally, we call a local access algorithm \textbf{robust} if it is able to answer adaptive queries according to the correct distribution.
\end{definition}
For the remainder of the paper all presented algorithms will be \emph{robust} (as they will succeed with high probability over any sequence of queries), and we only consider the weaker notion in the context of lower bounds. 

There are a few subtleties with the definition. The first is that even in the non-adaptive case, future queries may depend on the vertices returned by the algorithm. The $i$th query of $T$ is thus a function of $v_{t_1},\ldots,v_{t_{i-1}}$ (in the non-adaptive case) and the state of $\A$ after query $i-1$ (in the adaptive case). Therefore $D$ is defined recursively by conditioning on the first $i$ query responses. 
However, at the end of any query sequence with queried times $T$, the projection $P_T$ is clearly independent of the order elements in $T$ were queried, so the distribution $P_T(\U)$ is still well defined.

Finally, we state a basic result on partial sampling.
\begin{proposition}\label{prop:errsadd}
Let $G$ be a graph and $T$ an ordered list of determined times in a walk on $G$. Let $V_T$ be the associated set of determined positions. Suppose $V_T$ has been sampled to within $\ep$ of the true distribution in $\ell_1$ distance. For any new query $t$, let $t_-<t<t_+$ be the closest low and high previously determined times. These are denoted the \textbf{bracketing queries}. Then:
\begin{enumerate}
    \item The distribution of $v_t$ conditioned on $v_{t_-},v_{t_+}$ is equal to the distribution conditioned on all previously determined vertices.
    \item If $v_{t}$ is sampled from a distribution $D$ where $||D-P_{t-t_-}(D_C(v_{t_-},v_{t_+},t_+-t_-))||_1\leq \delta$, then $(V_T,v_{t})$ is $\ep+\delta$ close to the true distribution. Furthermore, if the true distribution of $v_{t}$ is some deterministic function of $k$ distributions, an equivalent result holds for sampling each distribution to within $\delta/k$ and returning the deterministic function applied to these samples. 
\end{enumerate}
\end{proposition}

In effect, this gives us the ability to only focus on bracketing queries while analyzing the closeness of a local access algorithm to uniform. 

\section{Local Access Via Spectral Expansion}\label{sec:expansion}
We first give an algorithm for undirected regular graphs with $\widetilde{O}(\frac{1}{1-\lambda}\sqrt{n})$ work per query. This is sublinear for small numbers of queries on graphs with good expansion, but is far from $\polylog$ work per query.
\begin{restatable}{theorem}{Regular}\label{alg:regular}
Fix $\ep>0$ and $B \in \N$ and $\lambda\geq 0$. Given $\rn$ and $\rv$ probe access to an undirected $d$-regular graph $G$ on $n$ vertices with $\lambda(G)\leq \lambda$, there is a deterministic local access algorithm which uses $O(\log(nt)\frac{1}{1-\lambda}\log(nB/\ep))$ additional space and $O(\sqrt{n}\cdot \log(nt)\frac{1}{1-\lambda}\log^2(nB/\ep))$ time and working space per query.
\end{restatable}
\begin{proof}
Let $k=O(\frac{1}{1-\lambda}\log(B n/\ep))$ be the smallest integer such that $\lambda(G^{k})\leq \ep/n^2B$.\\

The algorithm maintains a sorted list of previously determined positions (a superset of previous queries) and associated vertices $T=t_1<\dots <t_r$, $V_T=v_{t_1},\ldots,v_{t_r}$. Between queries, we maintain a constraint that for all $t_i,t_{i+1}$ it is either the case that $t_{i+1}-t_i=1$ (so the queries are direct neighbors) or $t_{i+1}-t_i\geq 2k$.\\

For a new query $t$, let $t_-\leq t<t_{+}$ be the bracketing queries, where $t_{+}=\infty$ if the constraint is unidirectional. For notational convenience, let $\rp(G,v,d)$ be a sequence of vertices obtain from making $d$ successive $\rn$ calls starting at vertex $v$. Furthermore let $l=t-t_-$ and $r=t_+-t$.
\begin{enumerate}
    \item\label{item:case1} If $l > 2k$ and $r>2k$, set $v_t=\rv(G)$.
    \item If $l \leq 2k$ and $r> 2k$, determine the vertices at $[t_{-}+1,t]$ as $\rp(G,v_{t_-},l)$. If $l> 2k$ and $r\leq 2k$, determine the vertices at $\{t_{+}-1,t_{+}-2,\ldots,t\}$ as $\rp(G,v_{t_+},r)$. 
    \item If neither condition is satisfied, we have $2k<|t_{+}-t_-|\leq 4k$. Let $d=\lfloor (t_{+}-t_-)/2\rfloor$ and let $L,R$ be empty sets of walks of length $d$ from $v_{t_-}$ and $v_{t_{+}}$ respectively. Let COL be the event a path from $L$ and $R$ share an endpoint.
    \begin{enumerate}
        \item Let $L \la L\cup \rp(G,d,v_{t_-})$.
        \item Let $R \la R\cup \rp(G,d,v_{t_{+}})$.
        \item If $COL$, go to Phase II. 
        \item After $2\sqrt{n}\log(B/\ep)$ iterations determine the vertices at $[t_-,t_{+}]$ as an arbitrary path, else repeat.
    \end{enumerate}
    In Phase II we have paths $p_l,p_r$ sharing an endpoint. If there are multiple colliding paths, choose the first to occur. Let the determined vertices $[t_-,t_{+}]$ be $p_l\bar{p_r}$ where $\bar{p}$ is the reverse of path $p$.
\end{enumerate}
Since the algorithm determines at most $2k$ timesteps per query (in Case III), the incremental persistent storage is at most $(\log(t)+\log(n))2k=O(\log(nt)\frac{1}{1-\lambda}\log(B n/\ep))$. The runtime is immediate from the description.

We now show the algorithm is $\ep$-close to uniform. Slightly abusing notation, let $T=t_1,\ldots,t_s$ be the ordered list of determined times \textit{after query $s$}. Let the $s+1$st query be denoted $t$ and write the bracketing vertices of $t$ as $t_-<t<t_+$. As before, let $l=t-t_-$ and $r=t_+-t$. By Proposition~\ref{prop:errsadd}, showing the distribution of the vertices decided at query $s+1$ are $\ep/B$-close to the true conditional distribution given the bracketing vertices $\{v_{t_-},v_{t_+}\}=\Phi$ suffices to show the algorithm is $\ep$ close via a union bound over the $B$ queries.
\begin{enumerate}
    \item If $t$ was decided in Case~\ref{item:case1}, we have $\Pr(v_t=v|\Phi)=1/n$. Let $W$ be the random walk matrix of $G$. Then
    \begin{align*}
        \Pr(P_{t}(D_C(G, \Phi)) = v) &= \frac{W^l_{v_{t_-},v}W^r_{v,v_{t_+}}}{\sum_{w\in V}W^l_{v_{t_-},w}W^r_{w,v_{t_+}}}
        \leq \frac{(1/n+\ep/B n)^2}{n(1/n-\ep/B n)^2}
        \leq \frac{1}{n}(1+O(\ep/B)).
    \end{align*}
    With a nearly identical lower bound. Taking a union bound over all $n$ potential vertices and adjusting $\ep$ by a constant factor completes the proof.
    \item If $t$ was decided in Case 2, we first decide the position at the \textit{end} of the random walk (abusing notation assume $t$ was this), and then assign the connecting walk. This is because fixing the endpoint $v_t$, the distribution of the decided path $v_{t_-} \ra v_t$ is clearly equal to the true distribution, since it was sampled via an unconstrained random walk. In the case where the walk was sampled $v_{t_+}\ra v_t$, since $G$ is undirected and regular the probability of a walk is equal to that of the reversed walk, so the decided path remains truly uniform. Then the analysis of deciding $v_t$ is nearly identical to Case I.
    
    \item If $t$ was decided in Case 3, since the distribution of left and right endpoints are $1/n^2$ close to uniform in $\ell_2$ distance, by a simple collision probability argument $\A$ fails to find a collision with probability $O(\ep/B)$ and loses an equal amount in $\ell_1$ distance. Otherwise, we first decide the position at the \textit{midpoint} of the random walk (abusing notation assume $t$ was this), and then assign the connecting walks. This is because fixing a collision (and thus endpoint) at $v_t$, the distribution of the decided paths $v_{t_-}\ra v_t$ and $v_t\ra v_{t_+}$ are equal to the true distribution, since they were sampled via unconstrained random walks.
    Then the analysis of deciding $v_t$ is nearly identical to Case I.
    \qedhere
\end{enumerate}
\end{proof}

\section{Random Regular Graphs}\label{sec:regular}
Next, we study the question of implementing access to \emph{random regular graphs}, which have the property that for all $d\geq 3$, the probability a random $d$-regular graph is an expander tends to 1. This implies that Theorem~\ref{alg:regular} composed with the set of random regular graphs achieves runtime $\widetilde{O}(\sqrt{n})$ per query. 
In fact, this is nearly the best possible runtime, as we prove no local access algorithm given probe access to random regular graphs making $o(\sqrt{n}/\log(n))$ probes per query achieves achieves subconstant error on \emph{adaptive} query sequences.
Furthermore, no local access algorithm making $o(n^{1/4})$ probes per query achieves subconstant error on \emph{non-adaptive} (in fact fixed in advance) query sequences. We first introduce notation for random $d$-regular graphs.
\begin{definition}
    Let $\Gd$ be the uniform distribution over $d$-regular graphs on $n$ vertices.
    \begin{itemize}
        \item For $d$ odd, we implicitly restrict to even $n$ when taking limits.
        \item For a set of edges $S=\{(v_1,w_1),\ldots,(v_k,w_k)\}$, let $\Gd\cap S$ be the uniform distribution over $d$-regular graphs on $n$ vertices containing all edges in $S$.
        Note that for certain $S$ (for instance, any containing a self-loop), this set is empty.
    \end{itemize} 
\end{definition}
For the remainder of the section we treat $d$ as a constant while $n$ trends to infinity, so $O$ notation sometimes hides factors dependent on $d$. Furthermore we assume $d\geq 3$, since the other two cases are degenerate. We now state informal versions of the main results. First, a sublinear algorithm for $\Gd$ obtained as a consequence of Theorem~\ref{alg:regular}.

\begin{restatable}{corollary}{RandomUB}\label{cor:random-ub}
There exists a deterministic local access algorithm $\A$ with time per query $O(\sqrt{n}\log^3(n))$ that, given $\rn$ and $\rv$ probe access to $\Gd$, satisfies for any adaptive query sequence $Q$ with $|Q|\leq \poly(n),$
$$\E_{G\la \Gd}||\D-P_{Q}(\U)||_1=o_n(1).$$
where $\D$ is the distribution of $\A$'s responses given probe access to $G$ over sequence $\Q$.
\end{restatable}
Next, an $\widetilde{\Omega}(\sqrt{n})$ lower bound against robust local access algorithms (i.e. those that face adaptive sequences).
\begin{restatable}[Informal Statement of Theorem~\ref{thm:lb-formal}]{theorem}{RandomLB}\label{thm:random-lb}
There is a constant $n_0$ and an adaptive sequence $\Q$ such that any robust local access algorithm $\A$ given $\rn$ and $\rv$ probe access to random $d$-regular graphs for $n\geq n_0$ with parameters $(\ep,B)=(.99, O(\log(n))$ makes $\Omega(\sqrt{n}/\log(n))$ graph probes per time query of $\Q$ in expectation.
\end{restatable}
Finally, an $\Omega(n^{1/4})$ lower bound that does not rely on adaptive query sequences.
\begin{restatable}[Informal Statement of Theorem~\ref{thm:oblivious-lb-formal}]{theorem}{ObliviousLB}\label{thm:oblivious-lb}
There is a constant $n_0$ and a fixed query sequence $Q$ such that any local access algorithm $\A$ given $\rn$ and $\rv$ probe access to random $d$-regular graphs for $n\geq n_0$ with parameters $(\ep,B)=(.99, n^{1/4})$ makes $\Omega(n^{1/4})$ graph probes per time query of $\Q$ in expectation.
\end{restatable}

It is impossible to prove lower bounds for \textit{all} subfamilies in $\Gd$ (in fact we give efficient local access algorithms for some later),
but any possible algorithm being $\Omega(1)$ from uniform on at least $99\%$ of random regular graphs effectively rules out a unified approach.

We begin by proving the $\widetilde{O}(\sqrt{n})$ upper bound using the algorithm from Section~\ref{sec:expansion}.
To do so, we recall the famous result that almost all random regular graphs are good expanders.
\begin{restatable}[\cite{Fri}]{lemma}{Exp}\label{lem:exp}
For all $d\geq 3$, $\Pr(\lambda(\Gd)\leq .95)=1-o_n(1).$
\end{restatable}
Then the proof follows directly.
\begin{proof}[Proof of Corollary~\ref{cor:random-ub}]
Choose $B=\poly(n)$, $\ep=1/\poly(n)$ and compose the algorithm of Theorem~\ref{alg:regular} with $\Gd$, where we promise that $\lambda\leq .95$. In the case of poorly expanding graphs this will result in walks that are arbitrarily far from truly random, but the runtime per query will still be as claimed.

For $G$ such that $\lambda(G)\leq .95$ we obtain that for any (potentially adaptive) query sequence $Q$, $||\D-P_Q(\U)||_1\leq 1/\poly(n)$. Then taking the expectation over $\Gd$ we obtain
\[\E_{G\la \Gd}||\D-P_{Q}(\U)||_1\leq 1/\poly(n)+\Pr[\lambda(\Gd)>.95]=o_n(1). \qedhere\]
\end{proof}

\subsection{Structure of Random Regular Graphs}
To prove the lower bounds, we first give three structural results which establish any algorithm must succeed even when the first $\Omega(\sqrt{n})$ graph probes define disjoint forests, and give tests for closeness of walks to the uniform distribution supported on only a few queries.

Our first goal is to show no algorithm making $\rn$ probes to $G\la \Gd$ can efficiently find cycles. This is essential, as the entire lower bound rests on the probes made by the algorithm defining a tree with $\Omega(1)$ probability. To do so, we first show conditioning on a small number of edges (e.g. those already known by the algorithm) does not increase the conditional probabilities of non-revealed edges by more than a constant factor.
\begin{restatable}{lemma}{Guessprob}\label{lem:guessprob}
For all $d \in \N$ there is a constant $c_d$ depending only on $d$ such that for an arbitrary set of edges $S$ with $|S|\leq \sqrt{n}$ and $v,w \in V$ arbitrary vertices where $(v,w) \notin S$, we have $\Pr_{G\la \Gd \cap S}[(v,w)\in G] \leq c_d/n$.
\end{restatable}
We defer the proof to Appendix~\ref{app:configmdl}. We use the configuration model of Bollobas~\cite{Bol} and a strengthening to handle degree sequences with small amounts of variation by~\cite{MW}.

Furthermore, probe access to $\Gd$ is equivalent to successively generating edges uniformly at random over the set of regular graphs satisfying the existing constraint - in effect, we can only determine edges when required, and this is the perspective we will use for the proof.
\begin{restatable}{lemma}{Marginalunif}\label{lem:marginalunif}
Let $\A$ be an algorithm having made $k$ arbitrary $\rn$ probes to $\Gd$ and let the returned edges be $E$. Then the conditional distribution over graphs given the probe responses is uniform over $\Gd\cap E$.
\end{restatable}
\begin{proof}
Let $v_1,\ldots,v_k$ the origin vertices for the $\rn$ probes and $w_1,\ldots,w_k$ the returned vertices. For $H \in \Gd\cap E$ we have $\Pr[\forall i, \rn_H(v_i)=w_i]=1/d^k$ whereas for $H \in \Gd\setminus E$ the equivalent probability is zero.
\end{proof}

Given these lemmas, we can now show the first $\Omega(\sqrt{n})$ probes made by any local access algorithm will fail to find cycles or merge forests with constant probability.
\begin{restatable}{lemma}{Randstruct}\label{lem:randstruct}
Let $\A$ be an algorithm, where at each step $\A$ makes a $\rn$ or $\rv$ probe to $\Gd$ or marks any vertex. Each vertex touched by a probe is marked. Then there is a constant $k_d$ depending only on $d$ such that for any $\A$ with at most $\sqrt{n}/k_d$ steps, with probability at least $.995$,
\begin{itemize}
    \item the $\rn$ probes will define a forest,
    \item no $\rn$ probe will ever merge two marked trees.
\end{itemize}
\end{restatable}
\begin{proof}
Let $V_{<i}$ be the set of vertices that are marked after probe $i-1$, and $E_{<i}$ the known edges. It is clear that the worst case is $\A$ making entirely $\rn$ queries. Let $v_{j}$ be the vertex queried at probe $j$. We have $|V_{<j}|\leq 2|E_{<j}|\leq 2j$. Define $k_d=20\sqrt{c_d}$ where $c_d$ is as in Lemma~\ref{lem:guessprob} and let $q=\sqrt{n}/k_d$. We obtain
\begin{align*}
    \Pr(\text{fail}) &\leq \sum_{i=1}^q \Pr_{\Gd \cap E_{<i}}(\rn(v_i) \in V_{<i})\\
    &\leq \sum_{i=1}^q |V_{<i}|\frac{c_d}{n}\\
    &= \frac{2c_d}{n}\frac{q(q+1)}{2}\\
    &\leq 1/200.
\end{align*}
Where the first line follows from Lemma~\ref{lem:marginalunif} and the second follows from Lemma~\ref{lem:guessprob}.
\end{proof}
\begin{restatable}{corollary}{Walktree}\label{cor:walktree}
For any $\ell\leq \sqrt{n}/\log(n)$, we have that $\Pr_{G\la \Gd}\Pr_{\sigma \la \U[\ell]}(\sigma \text{ defines a tree})\geq 1-O(1/\log^2(n)).$
\end{restatable}
\begin{proof}
This directly follows from setting $q=\sqrt{n}/\log(n)$ in the above proof, as a random walk is simply a sequence where at each step we probe $\rn$ at the current head.
\end{proof}
We now show random walks of length $\sqrt{n}/\log(n)$ over random regular graphs exhibit a distinguishing feature that can be checked on small segments. Intuitively, with high probability there will be no segment of length $r=\Omega(\log(n))$ where the simple path over the edges \textit{traversed in the walk} between the endpoints of the segment is shorter than $r/2$. Since the edges traversed by the walk will define a tree with high probability, an unusually short induced simple path implies the biased random walk corresponding to the tree metric in that segment is much shorter than its expectation, which is vanishingly unlikely. To show this, we formally define the \textbf{path length} of the induced simple path. 
\begin{definition}
For a partially determined vertex sequence $s=(s_1,\ldots,s_\ell) \in ([n],*)^\ell$,
let path length $\PL(s)$ be the distance between $s_1$ and $s_\ell$ in the induced (undirected, unweighted) graph $G'=([n],E')$,
where $(u,v) \in E'$ if and only if there exists $i$ such that $s_i=u,s_{i+1}=v$.
\end{definition}

We obtain that an unusually short simple path is vanishingly unlikely in any segment of a random walk.
\begin{restatable}{lemma}{Rand}\label{lem:rand}
Let $\sigma \in [n]^\ell$ be a walk of length $\ell\leq \sqrt{n}/\log(n)$. 
Let $F(\sigma)$ be the event any segment $s=(\sigma_i,\ldots,\sigma_j)$ of length $|s| \geq 40\log(n)$ has $\PL(s)<|s|/2$. Then 
$$\Pr_{G\la \Gd}\Pr_{\sigma\la \U[\ell]} [F(\sigma)]=o_n(1).$$
\end{restatable}
\begin{proof}
    Let $\Psi(\sigma)$ be the event $\sigma$ defines a tree. Then $\Pr_{G\la\Gd}\Pr_{\sigma\la\U[\ell]}(\Psi(\sigma)) \geq 1-O(1/\log^2(n))$ by Corollary~\ref{cor:walktree}.
    
    We now fix $G$ and sequentially generate a random walk $\sigma$. For each vertex in the random walk, there is some edge that was the first traversed by $\sigma$ (where we pick some edge for the first vertex arbitrarily). For each step of $\sigma$, at the current vertex $v$, label this first traversed edge a $-$ edge and all others $+$ edges. Then in a random walk in \textit{any} $d$-regular graph the probability of step $i$ being a $+$ step is exactly $(d-1)/d$ and these events are independent for all $i$. Furthermore, for all $\sigma$ that define a tree the $+$ and $-$ labels exactly correspond to the distance metric on the tree induced by the random walk, with $-$ corresponding to backtracking towards the initial vertex. 
    
    Now let $s$ be any segment of length at least $40\log(n)$. Let $F(s)$ be the event $F(\sigma)$ holds in this segment. Let $h(s)$ be the sum over $+$ and $-$ steps in $s$. Then by the definition of simple path and the correspondence between step labels and the tree metric:
    $$\{h(s) \geq s/2\} \cap \Psi(\sigma) \implies \overline{F(s)}.$$
    But then we can apply a basic Chernoff bound\footnote{Let $X_1,\ldots,X_n$ be independent random variables taking values in $\{0,1\}$. Let $X=\sum_{i=1}^nX_i$ and $\mu=\E[X]$. Then for any $\delta \in [0,1],$ $\Pr[X\leq (1-\delta)\mu] \leq \exp(-\delta^2\mu/2).$} to obtain $\Pr[h(s) < (1-\delta)\mu] \leq \exp(-\delta^2\mu/2)$. Choosing $\delta=1/4$ and using that $\mu =\E[h(s)]\geq 2s/3$ we obtain
    \begin{align*}
        \Pr[h(s) < s/2] &\leq \exp(-(2s/3)/32)\\
        &\leq n^{-1.2}
    \end{align*}
    Then taking a union bound over the at most $\ell^2$ such segments, we obtain 
    \begin{align*}
        \Pr(F(\sigma)) &\leq \Pr(\overline{\Psi(\sigma)})+\sum_{s\subseteq \sigma:|s|\geq 40\log(n)}\Pr(\{h(s) < s/2\})\\
        &\leq O(1/\log^2(n))+\ell^2\cdot n^{-1.2}\\
        &=o_n(1).\qedhere
    \end{align*}
\end{proof}

\subsection{Proof of Adaptive Lower Bound}
We are now prepared to prove the lower bound. For the remainder of the section let $\A$ be a local access algorithm with $\rn$ and $\rv$ probe access to $\Gd$.

We give a sequence of at most $c\log(n)$ time queries. By Lemma~\ref{lem:randstruct}, any algorithm that makes fewer than $\sqrt{n}/k_d c\log(n)$ probes per query sees non-merging trees with probability $.995$ for the duration of the query sequence. Given this occurs, we force the algorithm to return a walk segment that appears with probability $o(1)$ over the true distribution of random walks on $G$. 

We now begin to work with fixed instantiations of $\A$. We use the perspective of $\A$ successively determining the graph by making new $\rn$ probes.
\begin{definition}
    For a fixed instantiation of $\A$ on $\Gd$, let $T(Q)=(V_Q,S)$ be the \textbf{transcript} of the history of the algorithm after a sequence of queries $Q$. $V_Q$ holds the vertices returned at the times in $Q$, and $S$ holds the set of edges revealed by $\rn$ probes. Note the distribution over possible graphs at this time is $\Gd\cap S$.
\end{definition}
An adaptive query sequence is simply a function $f:T(Q)\ra \N$, where the next query is a (in our case deterministic) function of the existing transcript.
A non-adaptive query sequence is a function $g:V_Q\ra \N$, where the next query can only depend on the vertices returned by $\A$,
but not on the internal state of the algorithm.

\begin{notation}~
\begin{itemize}
    \item Given a queried time $t$, denote by $v_{t} \in V_Q$ the vertex returned by $\A$ for this time.
    \item Given a transcript $T(Q)=(V_Q,S)$, for vertices $v,w \in V$, let $\dist(v,w)$ be the length of the simple path between the vertices $v,w$ in the graph induced by the edges in $S$, where $\dist(v,w)=\infty$ if no path exists. Denote the simple path itself (if one exists) as $\sp(v,w)$.
\end{itemize}
\end{notation}

In the case where probes define non-merging trees, for all $v,w \in V$ once $\dist(v,w)<\infty$ it is fixed for the duration of the query sequence, and there are never multiple simple paths between vertices. This is a central component of the proof, as it implies $\A$ cannot \say{extend} paths without guessing.

We first give a family of distinguishing functions that we will use to lower bound $\ell_1$ distance, and show that truly random walks satisfy them with vanishing probability. The function $\F$ checks two conditions - if the \say{walk} traversed edges that do not actually exist, and if the path length of a sufficiently large segment of the walk is too short.
\begin{definition}
For an arbitrary graph $G=(V,E)$ let $\F:\{V,*\}^e\ra \{0,1\}$ be defined as
\[
\F(w_0,\ldots,w_e)
=\mathbb{I}\begin{cases}
\exists i \text{ st. } w_i \neq *, w_{i+1} \neq * \text{ and } (w_i,w_{i+1}) \notin E & \text{OR}\\
\exists i<j-40\log(n) \text{ st. } \PL(w_i,\ldots,w_j) < (j-i)/2
\end{cases}\]
Furthermore $\PL$ is nonincreasing (and thus $F$ is nondecreasing) with respect to revealing new vertices.
\end{definition}
Interestingly, the only reason we require knowing $G$ to define $F_G$ is to rule out edges that are not actually in the graph. 
\begin{remark}\label{lem:gap}
For $\ell \leq \sqrt{n}/\log(n)$ we have $\Eg\F(\U[\ell])=o_n(1)$ as a simple consequence of Lemma~\ref{lem:rand}.
Furthermore, as $\F$ is nondecreasing with regard to additional queries, for any set of timesteps $W \subseteq [\ell]$ and associated projection $P_{W}$ we obtain $\E_{G\la \Gd}\F(P_{W}(\U[\ell]))=o_n(1)$
\end{remark}

For our first lower bound, as we chose the next query time based on the transcript of $\A$ after the previous query,
we obtain that, for all local access algorithms, there \textit{exists} a sequence of bad queries.
Note that the algorithm of Theorem~\ref{alg:regular} succeeds asymptotically almost surely even on such adaptive sequences.
\begin{theorem}\label{thm:lb-formal}
There exist constants $q_d,n_0$ depending only on $d$, a family of distinguishing functions $\{F_G:G\in G(n,d)\}$,
and an adaptive query sequence $\Q$ of at most $O(\log(n))$ queries such that any (possibly randomized) local access algorithm $\A$,
given $\rn$ and $\rv$ probe access to $\Gd$ that makes fewer than $\sqrt{n}/q_d\log(n)$ probes per query satisfies for all $n\geq n_0$:
$$\E_{G\la \Gd}|\F(P_{\Q}(\U)) -  \F(\D)|\geq .99$$
where $\D$ is the distribution of $\A$'s responses given probe access to $G$ over sequence $\Q$.
\end{theorem}

\begin{proof}
Let $q_d=k_d\cdot 203$ where $k_d$ is from Lemma~\ref{lem:randstruct}. Our procedure generates a sequence of at most $203\log(n)$ queries, so by assumption $\A$ makes at most $\sqrt{n}/k_d$ probes. Thus by the lemma there is $n_1$ such that for $n>n_1$ with probability $.995$ the algorithm never finds cycles or merges trees. Note that we treat returned vertices as marked. 
Denote this event by $\Xi$, and for the remainder of the proof we assume it holds (and otherwise we can terminate the sequence).\\

Our first query is at time $e=\sqrt{n}/\log(n)$ (and there is an implicit query at time $0$). We claim either $\dist(v_0,v_e)=\infty$ or $\dist(v_0,v_e)< e/20$. Otherwise we would have 
$$\infty>\dist(v_0,v_e)\geq e/20=\sqrt{n}/20\log(n),$$
so the algorithm made at least $\sqrt{n}/20\log(n)$ probes at the first query, violating our assumption on probe complexity. 

If $\dist(v_0,v_e)< e/20$, we apply Lemma~\ref{rec:tooshort}. Thus we can extend the query sequence $\Q\la (\Q,Q')$ by at most $201\log(n)$ queries such that any returned transcript $T(\Q)$ either satisfies $\F(V_\Q)=1$ for all $G$ (in which case we are done) or contains $v_{t},v_{t'} \in V_\Q$ such that $d(v_{t},v_{t'})>|t'-t|$.

Now we have $v_{t},v_{t'} \in V_\Q$ such that $d(v_{t},v_{t'})>|t'-t|$, so we apply Lemma~\ref{rec:nopath}. Thus we can extend the query sequence $\Q\la (\Q,Q')$ by at most $\log(n)$ queries such that any returned transcript $T(Q)$ contains $v_{t},v_{t+1} \in V_Q$ such that $d(v_{t},v_{t+1})>1$. Then let $S$ be the edges in the transcript at the termination of the query sequence. We have $|S|\leq \sqrt{n}$ and so by Lemma~\ref{lem:guessprob}, and the definition of $F_G$,
$$\Pr_{G\la \Gd \cap S}[\F(V_Q)=1]\geq \Pr_{G\la \Gd \cap S}[(v_{t},v_{t+1}) \notin G]=1-o_n(1).$$ 
Then taking $n_2$ such that this term is at least $.999$, for $n>\max(n_1,n_2)$ we obtain $\Eg\F(\D)\geq .994$. Then by Remark~\ref{lem:gap} there exists $n_3$ such that for any projection $P_Q$, for all $n>n_3$ 
$$\Eg\F(P_{\Q}(\U[\ell]))\leq \Eg\F(\U[\ell])<.004,$$
and by taking $n_0=\max(n_1,n_2,n_3)$ the result follows.
\end{proof}

To complete the proof, we must give short query sequences that when $\Xi$ holds drive almost all distinguishing functions to 1. We first show an algorithm that does not know of a short enough path between returned vertices can be forced to return consecutive vertices \emph{in the walk} that it does not know a connecting edge between.
\begin{restatable}[No Viable Path Known]{lemma}{Nopath}\label{rec:nopath}
Assuming $\Xi$ holds, given a transcript $T(Q)$ suppose there are prior queries $v_x,v_y \in V_Q$ such that $\dist(v_x,v_y)>|y-x|$. Then there exists an adaptive extension of the sequence $Q'\la (Q,q)$ of at most $\log(n)$ queries such that for any returned transcript $T(Q')$ there are $v_t,v_{t+1}\in V_{Q'}$ such that $d(v_{t},v_{t+1})>1.$
\end{restatable}
\begin{proof}
We show this by binary searching on the \say{gap}. WLOG assume $x<y$. At each step:
\begin{enumerate}
    \item Query at time $m=\lfloor(x+y)/2\rfloor$.
    \item We have $\dist(v_x,v_m)+\dist(v_m,v_y)\geq  \dist(v_x,v_y)$ so by non-negativity either $\dist(v_x,v_m)>m-x$ or $\dist(v_m,v_y)>y-m$.
    \item If the first holds, let $y\la m$ and recurse. Otherwise let $x\la m$ and recurse.
\end{enumerate}
Since $y-x<\sqrt{n}$ at the start of the recursion after $\log(n)$ queries we drive $|x-y|$ to $1$, and so obtain $v_t,v_{t+1} \in V_Q$ such that $\dist(v_t,v_{t+1})>1$ as desired.
\end{proof} 
We next show algorithms cannot \say{fake} the existence of longer paths. The key idea is that modifying $\sp(v_0,v_e)$ (or finding a second simple path) after returning $v_e$ is impossible when the algorithm fails to find cycles. We force $\A$ to return vertices that either trigger Lemma~\ref{rec:nopath} or feature excessive backtracking, which drives the distinguishing function to 1.
\begin{restatable}[Known Path Too Short]{lemma}{Tooshort}\label{rec:tooshort}
Assuming $\Xi$ holds, given a transcript $T(Q)$ with $v_e \in V_Q$ suppose $\dist(v_0,v_e)<e/20$. Then there exists an adaptive extension of the sequence $Q'\la (Q,q)$ of at most $201\log(n)$ queries such that any returned transcript $T(Q')$ either contains $v_t,v_{t'} \in V_{Q'}$ where $d(v_{t},v_{t'})>|t-t'|$ or satisfies $\F(V_{Q'})=1$ for all $G$.
\end{restatable}
\begin{proof}
For the remainder of the analysis we implicitly assume that for all queries $t,t'$, $\dist(v_t,v_{t'})\leq |t-t'|$ since otherwise the transcript satisfies the first condition and we are done. We give a recursive construction of $q$ that \say{pushes down} the short path. Let $x\la 0,y\la e$.\\

At each step we maintain the invariants that $\dist(v_x,v_y)<(y-x)/10+20\log(n)+2$ and $200\log(n) \leq y-x$, which are initially satisfied by the lemma statement.
\begin{enumerate}
    \item Query $\A$ at time $m=\lfloor (x+y)/2\rfloor$.
    \item Let $r_m=\min_{v \in V}\{d(v_m,v): v \in \sp(v_x,v_y)\}$ be the length of the simple path from $v_m$ to the simple path from $v_x$ to $v_y$.
    \item If $r_m\geq 20\log(n)$, we apply Lemma~\ref{rec:spike} with $(x,y,m)$ which uses at most $3\log(n)$ additional queries and achieves the condition.
    \item If $r_m<20\log(n)$, we can bound the path length from some endpoint to $v_m$. Either $\dist(v_x,v_m) \leq \dist(v_x,v_y)/2+r_m$ or $\dist(v_m,v_y) \leq \dist(v_x,v_y)/2+r_m$. In the first case,
    \begin{align*}
        \dist(v_x,v_m) &\leq \dist(v_x,v_y)/2 +r_m\\
        &< ((y-x)/10+20\log(n)+2)/2+20\log(n)\\
        &\leq (m-x)/10+20\log(n) + 2
    \end{align*}
    so letting $y\la m$ the requirements of the recursion are satisfied. In the other case we set $x\la m$ and achieve the same.
    \item Then if $y-x< 200\log(n)$, we have $d(v_x,v_y)<(y-x)/20+20\log(n)+2 < (y-x)/2$. In this case, we query $\A$ at times $\{x+1,x+2,\ldots,y-2,y-1\}$.
    Then any set of vertices $\{v_x,\ldots,v_y\} \subset V_{Q'}$ where $d(v_t,v_{t+1})\leq 1$ for all $t$ lies entirely inside $S$, and thus must contain $\sp(v_x,\ldots,v_y)$. Therefore we have a walk segment of length at least $40\log(n)$ where $\PL(v_x,\ldots,v_y)= d(v_x,v_y)<(y-x)/2$ and thus $\F(V_{Q'})=1$ for all $G$ by the definition of $F_G$ as desired.
\end{enumerate}
Then the total number of queries is bounded above by $(1+200)\log(n)$ by inspection and Lemma~\ref{rec:spike}, so we conclude.
\end{proof}

\begin{restatable}{lemma}{Spike}\label{rec:spike}
Assuming $\Xi$ holds, given a transcript $T(Q)$ suppose there are $v_x,v_m,v_y \in V_Q$ where $\min_{v \in V}\{d(v_m,v): v \in \sp(v_x,v_y)\} \geq 20\log(n)$. Then there exists an adaptive extension of the sequence $Q'\la (Q,q)$ of at most $2\log(n)$ queries such that any returned transcript $T({Q'})$ either contains $v_t,v_{t'} \in V_{Q'}$ where $d(v_{t},v_{t'})>|t-t'|$ or satisfies $\F(V_{Q'})=1$ for all $G$.
\end{restatable}
\begin{proof}
As before, we assume that for all queries $t,t'$, the returned vertices $v_t,v_{t'}$ satisfy $\dist(v_{t},v_{t'})\leq |t-t'|$ since otherwise the transcript satisfies the first condition and we are done.

We have $x,m,y$ with a tree structure where the distance from $v_m$ to the simple path from $v_x$ to $v_y$ is at least $r_t\geq 20\log(n)$. Let
$$w = \argmin_{v \in V}\{d(v_m,v): v \in \sp(v_x,v_y)\}$$
be the vertex (which has not necessarily been returned) at the point where the simple path to $v_m$ branches from $\sp(v_x,v_y)$. With at most $\log(n)$ queries, we force $\A$ to output that the random walk visits $w$ at times $t_1\leq m-20\log(n)$ and $t_2\geq m+20\log(n)$.\\ 

To do so, we apply the following recursion. Let $a\leq b$ be times and $u$ a vertex where $u \in \sp(v_a,v_b)$.
\begin{itemize}
    \item Query $\A$ at time $t=\lfloor(a+b)/2\rfloor$. If $v_t=u$, halt.
    \item We have $\dist(v_a,v_t)\leq t-a$ and $\dist(v_t,v_b)\leq b-t$ by assumption.
    \item Either $u \in \sp(v_a,v_t)$ or $u \in \sp(v_t,v_b)$. If the first let $b\la t$ and otherwise $a\la t$.
\end{itemize}
After $\log(n)$ queries we drive $b-a$ to $1$. By assumption $d(v_a,v_b)\leq b-a=1$ and $u \in \sp(v_a,v_b)$, so $\A$ must have returned $u$ at some timestep.\\

We use this subrecursion twice, with $(a,b,u)=(x,m,w)$ for the first call and $(a,b,u)=(m,y,w)$ for the second. Let $t_1$, $t_2$ be the times obtained from these applications where $v_{t_1}=v_{t_2}=w$. We claim $t_1\leq m-20\log(n)$ and $t_2\geq m+20\log(n)$. If $t_1<m-20\log(n)$, we have $d(v_{t_1},v_m)=d(w,v_m)\geq  20\log(n)$ and thus $|t_1-m|<d(v_{t_1},v_m)$, violating our first assumption (and the other case is identical). But then if this does not occur, we have a segment $\{v_{t_1},\ldots,v_{t_2}\} \subset V_{Q'}$ of length at least $40\log(n)$ where $v_{t_1}=v_{t_2}=w$, so $\PL(v_{t_1},\ldots,v_{t_2})=0<40\log(n)/2$ which implies $\F(V_{Q'})=1$ for all $G$ as desired.
\end{proof}
This concludes the proof of our adaptive lower bound.

\subsection{Proof of Non-Adaptive Lower Bound}
\label{app:oblivious-lb}
The Theorem~\ref{thm:lb-formal} lower bound constructs valid adaptive query sequences, but relies on looking at the edges known to $\A$ to choose the next query and so does not rule out non-robust local access algorithms for $\Gd$. We now give a weaker $\Omega(n^{1/4})$ lower bound that uses a global query sequence (not even depending on the returned vertices) that still suffices to rule out efficient local access by an exponential margin.

\begin{restatable}{theorem}{ObliviousLBformal}\label{thm:oblivious-lb-formal}
There exist constants $k_d,n_0$ depending only on $d$, a family of distinguishing functions $\{F_G:G\in G(n,d)\}$,
and a fixed query sequence $Q$ of $n^{1/4}$ queries such that any (possibly randomized) algorithm $\A$ given $\rn$ and $\rv$ probe access
that makes fewer than $n^{1/4}/k_d$ probes per query satisfies for all $n\geq n_0$:
$$\E_{G\la \Gd}|\F(P_{Q}(\U)) -  \F(\D)|\geq .99 $$
where $\D$ is the distribution of $\A$'s responses given probe access to $G$ over sequence $Q$.
\end{restatable}
\begin{proof}
Take $k_d$ as in Lemma~\ref{lem:randstruct}, and define the query sequence as $Q=(n^{1/4},2,3,\ldots,n^{1/4}-1)$. For convenience, define $e=n^{1/4}$. The distinguishing function is identical to before, so by Remark~\ref{lem:gap} there is $n_1$ such that for $n\geq n_1$ we have $\Eg\F(P_Q(\U))<.004$.

As $|Q|=n^{1/4}$, the number of probes made by $\A$ is bounded by $\sqrt{n}/k_d$, and so by Lemma~\ref{lem:randstruct} there is $n_2$ such that for all $n>n_2$, with probability $.995$ the algorithm never finds cycles or merges trees. Note that we treat returned vertices as marked. Denote this event by $\Xi$.

Given $\Xi$ holds, we claim that at the completion of the query sequence either $d(v_0,v_{e})=\infty$ or $d(v_0,v_{e})<n^{1/4}/2$. If this was not the case, since $\A$ cannot alter $d(v_0,v_{e})$ after the first query without finding cycles, $\A$ made at least $n^{1/4}/2>n^{1/4}/k_d$ probes after the first query which violates our assumption on probe complexity. 

Then let the transcript at the end of the sequence be $T(Q)=(V_Q,S)$, recalling $S$ is the edges revealed via probes.
\begin{enumerate}
    \item If there exist $v_t,v_{t+1} \in V_Q$ such that $\dist(v_t,v_{t+1})>1$, we have $|S|\leq \sqrt{n}$  and so by Lemma~\ref{lem:guessprob} and the definition of $F_G,$
    $$\Pr_{G\la \Gd \cap S}(\F(V_Q)=1)\geq \Pr_{G\la \Gd \cap S}[(v_{t},v_{t+1}) \notin G]=1-o_n(1).$$
    \item If this never occurred, the segment $\{v_0,\ldots,v_e\} \subseteq V_Q$ traverses only edges in $S$, so it must contain all edges in $\sp(v_0,v_e)$. Therefore $\PL(v_0,\ldots,v_e)= \dist(v_0,v_e)< n^{1/4}/2$ and $\F(V_Q)=1$ for all $G$.
\end{enumerate}
Then taking $n_0=\max(n_1,n_2,n_3)$ where $n_3$ is chosen such that the $1-o_n(1)$ term is above $.999$, the result follows.
\end{proof}

\section{Efficient Local Access for Abelian Cayley Graphs}\label{sec:abelian}
We now turn to classes of graphs with algebraic structure. We achieve efficient (i.e. runtime polylogarithmic in $n$) local access to the hypercube, $n$-cycle and spectral expanders. In Appendix~\ref{app:algebraic}, we achieve efficient local access for arbitrarily dense graphs via the tensor and Cartesian product.
This is comparable to the work of~\cite{BRY,GGN}, which (among many other results) give a local access algorithms for random walks with fixed start and end vertices on specific classes of graphs.
\begin{definition}
    For a group $\Gamma$ of order $n$ and $S\subseteq \Gamma$, the Cayley graph $G=\Cay(\Gamma,S)$ is the degree $|S|$ graph on $n$ vertices where for all $g\in \Gamma, e \in S$ we add the edge $(g,ge)$ with label $e$. We call $S$ the \textbf{generators} of $G$. We say the Cayley graph is \textbf{abelian} if the subgroup generated by $S$ is. 
\end{definition}
More concretely, a Cayley graph is abelian if for all $e_i,e_j \in S$ we have $e_ie_j=e_je_i$. We do not require $S$ to be closed under inverses.

\begin{restatable}{theorem}{Abelian}\label{alg:abelian}
    Fix $\ep>0$ and $B \in \N$. Let $G=\Cay(\Gamma,S)$ be an abelian Cayley graph on $n$ elements with $d=|S|$, where for all $g\in G$, $g^2$ is computable in $\polylog(n)$ time.  There is a local access algorithm using $O(d\log(t))$ additional space and $d\cdot \polylog(n,t,B/\ep)$ time and working space per query.
\end{restatable}
We defer the proof to Appendix~\ref{app:algebraic}. In the parallel model, \cite{Teng} gives an algorithm for efficient generation of random walks on all Cayley graphs. For a walk of length $t$, they sample $\sigma \in S^t$ and compute the $t$ prefixes $\{s_i\}_{i\in [t]}=\{\prod_{j=1}^i\sigma_j\}_{i\in [t]}$ in parallel. Unfortunately, even computing a single prefix of a product of generators in sequential sublinear time is not obviously possible without further restrictions. 

Although abelianness represents a strong algebraic assumption, Theorem~\ref{alg:abelian} immediately provides local access algorithms for several graph families of interest in computer science. 
\begin{restatable}{corollary}{Abelianaccess}\label{cor:abelianaccess}
~
\begin{enumerate}
    \item By considering $\Gamma=(\Z/2\Z)^d$ and taking $S=(e_1,\ldots,e_d)$ as the generating set, there is an efficient local access algorithm for random walks on the dimension $d$ hypercube for all $d$.
    \item By considering $\Gamma=\Z/n\Z$ and taking $S=(1,-1)$ as the generating set, there is an efficient local access algorithm for random walks on the $n$-cycle for all $n$.
    \item By Proposition 5 of~\cite{AR}, for all $m \in \N$ there is an explicitly constructible set $S_m$ where $|S_m|=O(m)$ such that $\Cay(\Z_2^m,S_m)$ has spectral gap $1/3$. Thus there is an efficient local access algorithm for random walks on a class of $\polylog$ degree expanders of size $2^m$ for all $m$.
\end{enumerate}
\end{restatable}

We remark that despite all constant-degree abelian Cayley graphs being poor expanders, efficient local access is easy to provide, while for well-expanding random-regular graphs we obtain a polynomial lower bound. This indicates sublogarithmic mixing time is not a determinative property for efficient local access.

\bibliography{references}

\appendix

\section{Proof of Lemma~\ref{lem:guessprob}}\label{app:configmdl}
We apply the configuration model of~\cite{Bol}, extended to sequences of degrees. In the configuration model, given a degree sequence $\mathbf{d}=(d_i)_{i\in [n]}$, we place $d_i$ half-edges at vertex $i$ and connect all half-edges with a random matching. In the case where $d_i=d$ for all $i$, if the graph induced by a random matching is simple, we produce a random draw from $\Gd$. In our case, we \say{remove} half edges that are already occupied by $S$, place a random matching on the remaining half edges, and show that if the induced graph is simple and does not duplicate edges in $S$, we obtain a random draw from $\Gd\cap S$. We can use then use this to analyze conditional edge probabilities.

We first recall a lower bound on the probability that such a random matching induces a simple graph. For a degree sequence $\mathbf{d}$, define $D=D(\mathbf{d})=\sum_{i=1}^n d_i$, $D_2 = \sum_{i=1}^n d_i(d_i-1)$ and $D_3 = \sum_{i=1}^n d_i(d_i-1)(d_i-2)$. Let $P(\mathbf{d})$ be the probability that a random matching on $\mathbf{d}$ has no loops or multiple edges. The forthcoming lemma assumes $\max_i d_i^3 = o(D)$ which clearly holds in our application.
\begin{restatable}[\cite{MW} Lemma 5.1]{lemma}{Simple}\label{app:simple}
\[P(\mathbf{d})\geq \exp(-\frac{D_2}{2D}-\frac{D_2^2}{4D^2}-\frac{D_2^2D_3}{2D^4}).\]
\end{restatable}

We can then apply this lemma to prove the main claim. We remark that the bound $|S|\leq \sqrt{n}$ be be improved to $|S|=o(n)$, with $c_d$ depending on when $|S|/n$ falls below some constant threshold.
\Guessprob*
\begin{proof}
Let $\mathbf{d}=(d_i)_{i\in [n]}$ be the sequence where $d_i$ is the remaining degree of vertex $i$ given $S$. We place a random matching on this degree sequence. Given such a matching $M$, we contract it to a (multi) graph $G_M$ by treating each bucket as a single vertex.
\begin{restatable}{claim}{Unif}\label{clm:unif}
Given a randomly drawn matching $M$ where $G_M$ is simple and $G_M\cap S=\emptyset$, $G_M\cup S$ is a uniform draw from $\Gd\cap S$.
\end{restatable}
\begin{proof}
All possible simple graphs $G_M$ are induced by exactly $\prod_i(d_i!)$ matchings, so the conditional distribution over such graphs is uniform. Then multiplying by the indicator variable $\I[G_M\cap S=\emptyset]$, which corresponds to there being no duplicated edges between $G_M$ and $S$, produces the uniform distribution over the desired subset of graphs.
\end{proof}
We next show $G_M$ satisfies the conditions of Claim~\ref{clm:unif} with probability depending only on $d$. First, we show the matching is simple not considering the edges of $S$ with constant probability.
\begin{claim}
We have $\Pr(\I[G_M \text{ simple}]) = P(\mathbf{d}) \geq \exp(-d(d+2))$.
\end{claim}
\begin{proof}
We use the (crude) bounds $D\geq dn/2$, $D_2\leq d^2n$ and and $D_3\leq d^3n$.  Then applying Lemma~\ref{app:simple},
$$P(\mathbf{d}) \geq \exp(-d^2n/dn - d^4n^2/d^2n^2-d^4n^2d^3n/8d^4n^4) = \exp(-d(d+1+d^2/8n))$$
and choosing $n\geq d^2$ gives the claimed bound.
\end{proof}

We then show $G_M$ duplicates edges in $S$ with vanishing probability, which suffices to establish a constant lower bound on the probability of a \say{good} draw.
\begin{claim}
$\Pr(\{G_M \text{ simple}\} \cap \{G_M\cap S=\emptyset\})=\rho_d>0$
\end{claim}
\begin{proof}~
\begin{enumerate}
    \item Taking $n$ large enough $\Pr(G_M \text{ simple})\geq \exp(-d(d+2))$ by the previous claim.
    \item The probability of an edge between any two vertices in $G_M$ is at most $2d^2/dn$ by a union bound. There are at most $\sqrt{n}$ pairs of vertices with edges in $S$, so by a further union bound all such pairs are missing with probability at least $1-2d/\sqrt{n}$.
\end{enumerate} 
Then taking $n$ large enough that the second term is at least $1-\exp(-d(d+2))/2$, we have $\Pr(\I[G_M \text{ simple}] \cap \I[G_M\cap S=\emptyset])\geq \exp(-d(d+2)/2)/2=\rho_d$ as desired.
\end{proof}
Now we are almost done. We have $\Pr[(v,w)\in G_M]\leq 2d/n$ and thus
\begin{align*}
    \Pr_{G\la\Gd\cap S}[(v,w) \in G] &= \Pr[(v,w)\in G_M|\{G_M \text{ simple}\}\cap \{G_M\cap S=\emptyset\}]\\
    &\leq \Pr[(v,w)\in G_M]/\Pr[\{G_M \text{ simple}\} \cap \{G_M\cap S=\emptyset\}]\\
    &\leq \frac{2d/\rho_d}{n}
\end{align*}
So taking $c_d=2d/\rho_d$ (and increasing as needed to handle the small $n$ cases by making the bound greater than 1) we conclude.
\end{proof}

\section{Local Access With Algebraic Structure}\label{app:algebraic}
We now detail the approach for local access to abelian Cayley graphs and graph products. The methods we use are simple and similar to those of~\cite{Teng}, who construct algorithms for efficient \textit{parallel} generation of random walks on a variety of structured graphs. In each case, there is some element of algebraic structure that enables sampling the relevant feature of a walk (its position at a new timestep) via sampling lower-dimensional distributions.\\

We first recall the Multinomial ($\MNom$) and Multivariate Hypergeometric ($\MHGeom$) distributions, which we can sample from efficiently.
\begin{proposition}[\cite{BRY} Theorem 21] \label{prop:sampling}
Given $\ep>0$, we can sample from the following distributions within $\ep$ in $\ell_1$ distance:
\begin{enumerate}
    \item given $t \in \N$, $(p_1,\ldots,p_d) \in \mathbb{Q}^d$, we can generate $S\la \MNom(t,(p_1,\ldots,p_d))$ in time $O(d\cdot\polylog(t,1/\ep))$,
    \item given $m \in \N$, $(c_1,\ldots,c_d) \in \N^d$, we can generate $S\la \MHGeom(m,(c_1,\ldots,c_d))$ in time $O(d\cdot\polylog(m,\sum_i c_i,1/\ep))$.
\end{enumerate}
\end{proposition}
Note that sample time is linear in the dimension of the distribution but (poly)logarithmic in the number of elements.
\subsection{Low-Degree Abelian Cayley Graphs}

For all Cayley graphs, sampling a walk of length $\ell$ is equivalent to sampling a random product of elements in $S$ of length $\ell$. But in the abelian case, the value of a random product (and thus endpoint of a random walk) only depends on the \emph{counts} of elements in the product. Thus we can sample the distribution of edge labels, and thus endpoints, in time linear in $d$ but logarithmic in $\ell$. 

To do this, we first recall the distribution of edge labels in a random product. 
\begin{proposition}\label{prop:abStruct}
Let $G=\Cay(\Gamma,(e_1,\ldots,e_d))$ be an abelian Cayley graph where $|\Gamma|=n$.
\begin{enumerate}
    \item The counts of edge labels in a random walk of length $\ell$ from any vertex are distributed $\MNom(\ell,(1/d,\ldots,1/d))$.
    \item Let $D_C(c_1,\ldots,c_d)$ be the set of random walks from any vertex of length $\ell=\sum_{i=1}^dc_i$ that traverse $c_i$ edges with label $i$. Then the counts of edge labels along the first $t\leq \ell$ steps of walks in $D_C(c_1,\ldots,c_d)$ are distributed $\MHGeom(t,(c_1,\ldots,c_d))$.
\end{enumerate}
\end{proposition}

We can then provide local access to abelian Cayley graphs. In the random regular graph case, the difficulty came from sampling conditional \say{products}, but here we take advantage of that fact that permuting the order of elements in a product preserves endpoints in order to sample counts of edge labels unconditionally.

\Abelian*
\begin{proof}
The algorithm maintains a sorted list of previous times and positions $T=t_1<\dots <t_r$, $V_T=v_{t_1},\ldots,v_{t_r}$ where $t_1=0$ and $v_{0}$ is fixed at initialization. In addition, for all $i\in [r-1]$, the algorithm maintains a dictionary $L$, where $L_{t_i}=[l_1,\dots, l_d]$ with the invariant that the the walk has $l_j$ steps with label $j$ between $t_i$ and $t_{i+1}$. Given a dictionary entry $L_t=[l_1,\ldots,l_d]$ and $v \in G$, define  $G[v,L_t] = v\prod_{e_i\in S}e_i^{l_i}$.

\noindent
Given a new query $t$:
\begin{enumerate}
    \item If $t>t_r$, set $L_{t_r} \la \MNom(t-t_r,(1/d,\ldots,1/d),\ep/B)$, and set $v_{t} \la G[v_{t_r},L_{t_r}]$. 
    \item Otherwise let $t_-<t<t_+$ be the bracketing queries. Sample $D\la \MHGeom(t-t_-, L_{t_-}, \ep/B)$, set $v_{t} \la G[v_{t_-},D]$, set $L_{t} \la L_{t_-}-D$ and set $L_{t_-}\la D$.
\end{enumerate}
Storing the query time takes incremental space $O(\log(t))$, storing the determined vertex $O(\log(n))$, and storing the dictionary $O(d\log(t))$.

The runtime is immediate from Proposition~\ref{prop:sampling} and the assumption that group products are computable in time $\polylog(n)$, so we can use $d$ iterations of repeated squaring with each requiring time $\polylog(t,n)$. 

In both the unidirectionally and bidirectionally constrained case, the vertex reached by a random walk on an abelian Cayley graph is a deterministic function of the counts of the bracketing \emph{edge labels}. Therefore ensuring the counts in each new dictionary are sampled to within $\ep/B$ of the true distribution is sufficient to establish the approximation by Proposition~\ref{prop:errsadd}. Since in both cases we approximate the true distribution to within $\ep/B$ in $\ell_1$ distance by Proposition~\ref{prop:abStruct}, the result follows.
\end{proof}

\subsection{Graph Products}\label{subsec:product}
We can utilize the structure of common graph product operations to provide local access, given algorithms for their components. 
To do so, we give arguably the simplest possible local access algorithm, one that is only efficient when the time queries are far larger than the size of the graph, to use as the basis for product constructions.
\begin{restatable}{lemma}{CLSbase}\label{alg:CLS-base}
Fix $\ep>0$ and $B \in \N$. Given a graph $G=(V,E)$ on $n$ vertices with $\lambda(G)<c$, there is a local access algorithm, which uses $O(\log(tn))$ additional space and runtime $O(\poly(n,\log(t/\ep))$ time and working space per query.
\end{restatable}
\begin{proof}
The algorithm solely remembers previously determined times and vertices $v_{t_1},\ldots,v_{t_k}$. Given query $t$, let $t_-<t<t_+$ be the bracketing times. For convenience, define $l=t_+-t_-$ and $m=t-t_-$. We then explicitly sample the desired distribution to within $\ep$ in $\ell_1$ distance. Let $W$ be the transition matrix of $G$. Then for $v\in V$, $$\Pr(P_{m}(D_C(G,v_{t_-},v_{t_+},l))=v)=W^{m}_{v_{t_-},v}W^{l-m}_{v,v_{t_+}}/\sum_{u \in V}W^{m}_{v_{t_-},u}W^{l-m}_{u,v_{t_+}}.$$
We can then use $n\log(t)$ repeated squares of the transition matrix to compute the PDF, and then sample to the desired accuracy and return.
\end{proof}
To make this algorithm concrete, for an undirected aperiodic graph $G$ on $n$ vertices with $\ep=n^{-c}$, we obtain a runtime of $\widetilde{O}(n^\omega)$, while we desire runtime polylogarithmic in $n$. 

We first examine the tensor product of graphs.
\begin{definition}
    Given graphs $G_1=(V_1,E_1), G_2=(V_2,E_2)$ the \textbf{tensor product} of $G_1$ and $G_2$, denoted $G_1\times G_2$, is the graph with vertex set $V_1\times V_2$ where $(v_1,v_2),(w_1,w_2)$ are adjacent if and only if $(v_1,w_1) \in E_1$ and $(v_2,w_2) \in E_2$.
\end{definition}
The projection of a random walk on the tensor product onto its component graphs is an independent random walk over each graph. Then we can easily decompose sampling conditional products to sampling on the components.
\begin{restatable}{lemma}{Tensor}\label{alg:tensor}
    Given local access algorithms $\A_1,\A_2$ for graphs $G_1,G_2$ running in time $T(\A_1,\ep,B,t)$, $T(\A_2,\ep,B,t)$, there is a local access algorithm $\A_T$ for $G_1\times G_2$ with runtime $T(\A_T,\ep,B,t)=T(\A_1,\ep/2,B,t)+T(\A_2,\ep/2,B,t)+O(\log(|G_1|\cdot|G_2|,t,B/\ep))$.
\end{restatable}
\begin{proof}
The algorithm initializes both sub-algorithms $\A_1,\A_2$ with parameters $\ep/2,B$. Upon receiving query $t$, $\A_T$ itself queries $\A_1,\A_2$ with time $t$. Let the obtained vertices be $v',w'$ respectively, and $\A$ returns $(v',w')$. Since the vertex in a walk on a tensor product is a deterministic function of the two (independent) component distributions, by Proposition~\ref{prop:errsadd} we obtain the desired approximation. The runtime is composed of the required calls to the sub-algorithms, plus the time to write the inputs to each and output the returned vertex.
\end{proof}

We then obtain efficient local access to walks on arbitrarily dense graphs.
\begin{restatable}{corollary}{Denseaccess}\label{cor:denseaccess}
Fix $\ep>0$ and $B \in \N$. Let $G$ be an arbitrary graph. For all $k\geq 1$ there is a local access algorithm for $G^{\times k}$ with runtime $O(k\log^2(B/\ep))$, where we hide factors polynomial in $|G|$. 
\end{restatable}
\begin{proof}
Let $b(i)\in \{0,1\}^{\log(k)}$ be the representation of $i$ in binary. Then $G^{\times k}\cong \bigtimes_{j \in b(k)}G^{\times 2^j}$ and so the iterated tensor product can be written as a binary tree of products with depth bounded by $2\log(k)$. Choosing the uppermost local access algorithm algorithm to have error $\ep/B$ implies the leaves have error parameter $\ep/B2^{2\log(k)}=\Omega(\ep/B\log^2(n))$. Then the result follows by applying the algorithm from Lemma~\ref{alg:CLS-base} to the $k$ copies of $G$ at the leaves of the tree. 
\end{proof}
For $G$ a regular graph with degree $d$, since $|G^{\times k}|=|G|^k$ and  $\deg(G^{\times k})=d^k$, we obtain efficient local access to infinite families of degree $(|G|^k)^\delta$ graphs for any $\delta=\log_{|G|}(d) \in (0,1]$. This indicates that high degree does not prevent efficient local access.

Walks on the Cartesian product similarly have a component decomposition (although this time we restrict to regular graphs).
\begin{definition}
    Given regular graphs $G_1=(V_1,E_1), G_2=(V_2,E_2)$ the \textbf{Cartesian product} of $G_1$ and $G_2$, denoted $G_1\square G_2$, is the graph with vertex set $V_1\times V_2$ where $(v_1,v_2),(w_1,w_2)$ are adjacent if and only if 
    \begin{itemize}
        \item $v_2=w_2$ and $(v_1,w_1) \in E_1$ or,
        \item $v_1=w_1$ and $(v_2,w_2) \in E_2$.
    \end{itemize}
\end{definition}
In a similar manner to the abelian Cayley case, we use the ability to decompose walks of length $r$ as $t$ and $r-t$ steps on the first and second coordinate respectively to always sample counts with a unidirectional constraint. Conditioning on these counts, we can sample efficiently given the local access algorithms for each component. 

\begin{restatable}{lemma}{Cartesian}\label{alg:cartesian}
Given local access algorithms $\A_1,\A_2$ for regular graphs $G_1,G_2$ running in time $T(\A_1,\ep,t)$, $T(\A_2,\ep,t)$, there is a local access algorithm $\A_C$ for $G_1\square G_2$ running in time $T(\A_C,\ep,t)=T(\A_1,\ep/2,t)+T(\A_2,\ep/2,t)+\polylog(|G_1|\cdot|G_2|,t,B/\ep)$.
\end{restatable}
\begin{proof}
The approach is similar to that of Lemma~\ref{alg:abelian}. Rather than sample a conditional walk at each step, for a unidirectionally constrained walk of length $\ell$ we sample the number of steps on each component, and use these \say{times} as inputs to $\A_1$ and $\A_2$.

Initialize $\A_1,\A_2$ with $\ep=\ep/3$, $B=B$. 
Let $d_1=\deg(G_1)$ and $d_2=\deg(G_2)$. The algorithm maintains a sorted list of previously queried times and positions $T=t_1,\ldots,t_r$, $V_T=v_{t_1},\ldots,v_{t_r}$ where $t_1=0$ and $v_{0}$ is fixed at initialization. In addition, the algorithm maintains $S=s_{t_1},\ldots,s_{t_{r}}$ where $s_i$ is the number of steps on $G_1$ in the interval $[0,t_{i}]$. Given a new query $t$:
\begin{enumerate}
    \item If $t>t_r$ set 
    $s_{t}\la s_{t_r}+\BNom(t-t_r,(d_1/(d_1+d_2),d_2/(d_1+d_2)),\ep/3B).$
    \item Otherwise let $t_-<t<t_{+}$ be the bracketing queries. Set $s_{t}\la s_{t_-}+\HGeom(t-t_-,(s_{t_+}-s_{t_-},t_+-t_-), \ep/3B)$.
\end{enumerate}
Finally set $v_{t}\la (\A_1(s_{t}),\A_2(t-s_{t}))$.

The runtime consists of sampling via Proposition~\ref{prop:sampling}, calling the algorithms for components and writing the output to the tape. The analysis of closeness in distance is nearly identical to that of Theorem~\ref{alg:abelian}, except that we take $\ep\la\ep/3$ as the final vertex is a function of three sampled distributions.
\end{proof}
In an identical manner to the construction of higher tensor powers, composing Lemma~\ref{alg:cartesian} with itself in a binary tree and using Lemma~\ref{alg:CLS-base} as a base case gives local access to the $k$th cartesian product $G^{\square k}$ of a $d$-regular graph $G$ with runtime $O(k\polylog(\ep/B))$, again hiding factors in $|G|$.

\end{document}